\documentclass[11pt]{article}
\usepackage[T1]{fontenc}
\usepackage{lmodern}
\usepackage[margin=1in]{geometry}
\usepackage{microtype}
\usepackage{amsmath,amssymb,amsthm,mathtools}
\usepackage{mathrsfs}
\usepackage[colorlinks=true,linkcolor=blue,citecolor=blue,urlcolor=blue]{hyperref}
\hypersetup{pdftitle={Polynomial Separation of a Shared Tensor Leg and Matrix Multiplication in Arbitrary Characteristic}}
\newtheorem{theorem}{Theorem}[section]
\newtheorem{proposition}[theorem]{Proposition}

\newtheorem{corollary}[theorem]{Corollary}
\theoremstyle{definition}

\theoremstyle{remark}

\newcommand{\eps}{\varepsilon}
\newcommand{\degen}{\mathrel{\rightsquigarrow}}

\DeclareMathOperator{\rank}{R}

\numberwithin{equation}{section}
\title{Polynomial Separation of a Shared Tensor Leg\\
and Matrix Multiplication in Arbitrary Characteristic}
\author{Houxi L\"u}
\date{October 9, 2026}
\begin{document}
\maketitle
\begin{abstract}
We improve the shared-leg separation construction in Proposition 3.1
of OpenAI, \emph{An Upper Bound of $9/4$ for the Matrix Multiplication
Exponent} (October 2, 2026). If $M$ tensor blocks have a common first
space and independent matched sectors on their other two spaces, then
$M+n-1$ copies of their matched sum degenerate to the full direct sum
of the blocks, each tensored with an $n$-dimensional dot product.
For $n=M$, this reduces the number of source copies from $5M$ to
$2M-1$. The construction uses polynomial interpolation and truncated
local substitution, with no roots of unity or factorial denominators.
It works over every field containing enough distinct elements.
We verify that the remaining spectral and discrete-growth arguments
of the cited preprint are independent of characteristic. Over an
algebraic closure, the replacement therefore gives $\omega\leq9/4$
in every characteristic.
\end{abstract}

\section{The replacement and its setting}

The separation problem is to make each block read its own copy of
a shared input. In Proposition 3.1 of
OpenAI~\cite{OAI:Matrix-Multiplication-Nine-Fourths-October-2-2026},
this is done by Fourier projection followed by a weight degeneration.
We use evaluation of polynomials instead. Evaluation at $M$ distinct
points gives $M$ independent input coordinates; a local substitution
allows us to approach this evaluation map through restrictions of
ordinary polynomial multiplication.

We follow the notation of the cited paper: $A=\sum_h A_h$ denotes
a matched sum, $A\geq B$ means that $B$ is a restriction of $A$,
$\otimes$ denotes tensor product, and $A^{\oplus r}$ denotes a full
direct sum of $r$ copies. An integer $r$ also denotes the direct
sum of $r$ scalar tensors, so $A\otimes r\cong A^{\oplus r}$.
All spaces are finite-dimensional over a algebraically closed field $F$. Tensors are
viewed as trilinear forms, and a restriction pulls a form back
along a linear map on each of its three spaces.

We use $\degen$ for an algebraic degeneration: the three pullback
maps may have entries in $F[\eps,\eps^{-1}]$, provided that the
resulting form is polynomial in $\eps$; its value at $\eps=0$ is
the target tensor. This extends the fixed-substitution, monomial-weight
convention of Section 2.2 of the cited paper. Its interpolation
argument still applies. Indeed, if the resulting polynomial family
has degree $L$ and constant term $B$, interpolation of its $j$th
tensor power at $jL+1$ distinct nonzero points gives
\[
 (jL+1)\lambda(A)^j\geq\lambda(B)^j.
\]
Over an infinite field, taking $j$th roots proves degeneration
monotonicity for every tensor character $\lambda$.

\section{Separation by polynomial multiplication}

\begin{proposition}[Shared-leg separation]\label{prop:separation}
Let $M,n\geq1$, and suppose that $F$ contains at least $M+n-1$
distinct elements. Let
\[
 A_h:X\times Y_h\times Z_h\longrightarrow F
 \qquad(1\leq h\leq M)
\]
be trilinear forms, and define their matched sum by
\[
 A(x,y,z)=\sum_{h=1}^M A_h(x,y_h,z_h),\qquad
 y\in\bigoplus_hY_h,\quad z\in\bigoplus_hZ_h.
\]
Then
\begin{equation}\label{eq:separation}
 \boxed{
 A^{\oplus(M+n-1)}\degen
 \bigoplus_{h=1}^M\bigl(A_h\otimes B_X(n)\bigr).
 }
\end{equation}
Here $B_X(n)$ is the dot-product tensor
\[
 B_X(n):F\times J\times J^\vee\longrightarrow F,
 \qquad(c,j,\varphi)\longmapsto c\varphi(j),
\]
for any $n$-dimensional space $J$. The direct sum on the right of
\eqref{eq:separation} has independent spaces on all three legs.
\end{proposition}

\begin{proof}
We want the $h$th block to read only the $h$th copy of $X$.
Polynomial interpolation provides exactly such coordinates.
Choose distinct $a_1,\ldots,a_M\in F$ and put
\[
 P=F[t]_{<M},\qquad J=F[u]/(u^n),
\]
where $F[t]_{<d}$ denotes the polynomials of degree less than $d$.
The map
\begin{equation}\label{eq:evaluation}
 e:P\longrightarrow F^M,\qquad
 p\longmapsto\bigl(p(a_1),\ldots,p(a_M)\bigr)
\end{equation}
is an isomorphism: a polynomial of degree less than $M$ cannot
vanish at all $M$ points unless it is zero, and the dimensions agree.
Thus $X\otimes P$ can supply the desired independent copies of $X$.

For each $h$, substitute $t=a_h+\eps u$ and define
\[
 \sigma_h(p,j)=p(a_h+\eps u)j\quad\text{in }J[\eps].
\]
The substitution is defined on all of $F[t]$; in particular,
$\sigma_h$ is bilinear on $P\times J$. It has two useful properties.
At $\eps=0$ it becomes $\sigma_h(p,j)=p(a_h)j$, which reads
exactly the evaluation coordinate in \eqref{eq:evaluation}.
For $\eps\ne0$, substitution identifies $F[t]_{<n}$ with $J$:
\[
 (t-a_h)^s\longmapsto\eps^s u^s\qquad(0\leq s<n).
\]
Both sides are bases, and all the displayed factors are invertible.

We now assemble these actions into one polynomial family on
\[
 \widetilde X=X\otimes P,\qquad
 \widetilde Y=\bigoplus_h(Y_h\otimes J),\qquad
 \widetilde Z=\bigoplus_h(Z_h\otimes J^\vee).
\]
Define $\Phi(\eps)$ on matching components by
\begin{equation}\label{eq:phi}
 \Phi(\eps)(x\otimes p,y_h\otimes j,z_h\otimes\varphi)
 =A_h(x,y_h,z_h)\varphi\bigl(\sigma_h(p,j)\bigr),
\end{equation}
and set it to zero on components with different $Y$- and $Z$-labels.
Extend by trilinearity. Each coefficient is polynomial in $\eps$.
We will obtain this family by restricting $M+n-1$ copies of $A$.

Put $\ell=M+n-1$, $Q=F[t]_{<n}$ and $W=F[t]_{<\ell}$.
Let
\[
 \mu:P\otimes Q\longrightarrow W,\qquad \mu(p,q)=pq
\]
be ordinary multiplication, and let
\[
 C_\mu(p,q,\psi)=\psi(pq)
 \qquad(p\in P,\ q\in Q,\ \psi\in W^\vee)
\]
be its trilinear form. For each $h$, define
\[
 g_h:Q\longrightarrow J[\eps],\quad q\longmapsto q(a_h+\eps u),
 \qquad
 \iota_h:W\longrightarrow J[\eps],\quad f\longmapsto f(a_h+\eps u).
\]
Here and below maps with polynomial coefficients are extended
linearly over the parameter ring when composed. Substitution
preserves products, so
\begin{equation}\label{eq:commutation}
 \sigma_h\circ(\operatorname{id}_P\otimes g_h)
 =\iota_h\circ\mu.
\end{equation}
This is the relation that connects our family to ordinary
multiplication.

Work temporarily over $F[\eps,\eps^{-1}]$, so that $g_h$ is
invertible. Pull back $A\otimes C_\mu$ along the three maps
\begin{align*}
 \alpha&=\operatorname{id}_{X\otimes P},\\
 \beta&=\bigoplus_h(\operatorname{id}_{Y_h}\otimes g_h^{-1}),\\
 \gamma&=\bigoplus_h(\operatorname{id}_{Z_h}\otimes\iota_h^\vee).
\end{align*}
The maps on the last two legs take
$\widetilde Y$ to $\bigoplus_h(Y_h\otimes Q)$ and
$\widetilde Z$ to $\bigoplus_h(Z_h\otimes W^\vee)$, respectively.
On the $h$th component, the auxiliary pairing becomes
\[
 (\iota_h^\vee\varphi)\bigl(\mu(p,g_h^{-1}j)\bigr)
 =\varphi\bigl(\sigma_h(p,j)\bigr)
\]
by \eqref{eq:commutation}. Hence the pullback is precisely
$\Phi(\eps)$, including its zero terms on unmatched components.

It remains to pay for $C_\mu$. Choose $\ell$ distinct points
$b_1,\ldots,b_\ell\in F$. Evaluation at these points identifies
$W$ with $F^\ell$. Multiplication in $W$ can therefore be computed
by evaluating the two inputs, multiplying their values at the
$\ell$ points, and interpolating the output. Thus
\[
 \ell\geq C_\mu,\qquad
 A^{\oplus\ell}\cong A\otimes\ell\geq A\otimes C_\mu.
\]
Composing this fixed restriction with $\alpha,\beta,\gamma$ gives
Laurent-polynomial pullback maps from $A^{\oplus\ell}$ whose
result is the polynomial family $\Phi(\eps)$.

Finally, at $\eps=0$, equation \eqref{eq:phi} reads
\[
 \Phi(0)(x\otimes p,y_h\otimes j,z_h\otimes\varphi)
 =p(a_h)A_h(x,y_h,z_h)\varphi(j).
\]
Under $\operatorname{id}_X\otimes e:X\otimes P\cong\bigoplus_hX$,
the $h$th block reads only the $h$th copy of $X$. Its other two
spaces are already separate. Consequently
\[
 \Phi(0)\cong\bigoplus_h\bigl(A_h\otimes B_X(n)\bigr),
\]
which proves the degeneration.
\end{proof}

The construction uses no roots of unity and no division by
factorials. The inverse of $g_h$ is simply
$j(u)\mapsto j((t-a_h)/\eps)$, using the representative of degree
less than $n$ in $J$. In particular, its Laurent-polynomial form
is valid in every characteristic. The family $\Phi$ has degree
at most $\min(M-1,n-1)$ in $\eps$, since powers $u^n$ and higher
vanish in $J$.

\section{The spectral comparison}

Let $F$ now be infinite, and let $\lambda$ be a tensor character
in the sense of Section 2 of the cited paper: it is normalized,
restriction-monotone, additive on full direct sums, and
multiplicative on tensor products. Write
$\lambda(B_X(n))=n^{p_X}$, with $p_X\in[0,1]$, as in its
Section 2.1. Proposition~\ref{prop:separation} and interpolation
monotonicity give
\begin{equation}\label{eq:finite}
 (M+n-1)\lambda(A)\geq n^{p_X}\sum_h\lambda(A_h).
\end{equation}
Taking $n=M$ yields
\begin{equation}\label{eq:comparison}
 (2M-1)\lambda(A)\geq M^{p_X}\sum_h\lambda(A_h).
\end{equation}
This replaces the factor $5M$ in Proposition 3.1 of the cited paper.

The proof of its Corollary 3.2 now goes through with this new
factor. To see precisely what changes, let $T=\sum_iT_i$ be a
matched sum and let $q$ be a rational probability vector. In
$T^{\otimes N}$ retain only sector words containing $Nq_i$
copies of label $i$. There are
\[
 M_N=\frac{N!}{\prod_i(Nq_i)!}
\]
blocks, each of character value $\prod_i\lambda(T_i)^{Nq_i}$.
Equation \eqref{eq:comparison} gives
\[
 \lambda(T)\geq
 \left(\frac{M_N}{2M_N-1}\right)^{1/N}
 M_N^{p_X/N}\prod_i\lambda(T_i)^{q_i}.
\]
The first factor tends to $1$, while $\log M_N/N$ tends to the
entropy $H(q)=-\sum_iq_i\log q_i$. By continuity in $q$, we
recover exactly the original shared-leg inequality
\begin{equation}\label{eq:entropy}
 \lambda(T)\geq e^{p_XH(q)}\prod_i\lambda(T_i)^{q_i}.
\end{equation}
The factorials here count words as integers; they are not
inverted in $F$. The improved finite copy count therefore
preserves the entropy exponent. By itself it does not give a
strict improvement below $9/4$.

\section{The consequence in positive characteristic}

The new separation argument works over every infinite field.
To apply the remainder of the cited paper in positive
characteristic, we need to check its other field-dependent steps.
We record these checks rather than repeat its proofs.

First, the rank bound for polynomial multiplication follows
from interpolation at distinct points, which is valid over any
infinite field. Second, the determinant filtration in its
Lemma 4.1 uses
\[
 0\longrightarrow (uw-vs)V_{e-1}
 \longrightarrow V_e\otimes\operatorname{span}(s,w)
 \longrightarrow V_{e+1}\longrightarrow0.
\]
Here $V_e$ is the space of homogeneous degree-$e$ forms in $u,v$,
and the last map substitutes $(s,w)=(u,v)$. It is surjective on
monomials. Multiplication by $uw-vs$ is injective, and its image
has the dimension of the kernel. These facts hold in every
characteristic, including $2$. Multiplication by a first input
preserves this kernel, so the block degeneration in that lemma
also holds without a characteristic assumption.

Third, the three-sector construction in its Lemma 4.2 only
uses coordinate restrictions, basis permutations, and integer
weights of a degeneration parameter. It does not divide by $3$:
the uniform probability vector and its entropy are real numbers
used to compare character values. Finally, the detecting-character
proof in its Appendix A uses the ordered tensor semiring,
flattening ranks, finite-dimensional real separation, compactness,
and a fixed-point theorem. Every nonzero tensor still restricts
to a scalar tensor over an arbitrary field, and all semiring and
rank comparisons used there remain valid. The separation and
fixed-point arguments concern real-valued states, not the
characteristic of the tensor field.

These observations justify using Sections 4 and 5 and Appendix A
of the cited paper over an algebraically closed field $F$ of any
characteristic, with Proposition~\ref{prop:separation} in place
of its Proposition 3.1. For clarity, the resulting chain of
inequalities is as follows. Put
\[
 t=(p_X+p_Y+p_Z)/3>0,\qquad
 C(a,b)=\sum_{i=0}^{a-1}\sum_{j=0}^{b-1}x_i y_j z_{i+j},
\]
and use the original symmetrized profile
\[
 P(a,b)=\left(\prod_{\pi\in S_3}\lambda_\pi(C(a,b))\right)^{1/(6t)}.
\]
Its concavity and shifted tripling follow from the two
constructions just checked and \eqref{eq:entropy}. The real
argument of Lemma 5.1 then gives
\[
 a^{4/3}\leq P(a,a)\leq(2a-1)^{1/t}.
\]
Consequently $t\leq3/4$, and every character satisfies
$\lambda(T_d)\leq d^{9/4}$ for the matrix multiplication tensor
$T_d=\sum_{i,j,k=1}^d x_{ij}y_{jk}z_{ki}$.

Write $\nu=\inf_{d\geq2}\log_d\rank(T_d)$, again as in the
cited paper. Its detecting-character lemma, applied with
$k=\lceil d^\nu\rceil-1$, yields
\[
 d^\nu\leq d^{9/4}+1.
\]
Letting $d\to\infty$ gives $\nu\leq9/4$. The recursive
bilinear-algorithm argument in its proof of Theorem 1.1 therefore
gives $\omega_F\leq9/4$.

\begin{corollary}\label{cor:fields}
For every algebraically closed field $K$, the matrix multiplication exponent satisfies
$\omega_K\leq9/4$. In particular, the bound holds in all but
finitely many characteristics; in fact no characteristic needs
to be excluded.
\end{corollary}

\section*{Acknowledgements}

I would like to thank GPT 6.1 Sol (OpenAI) and Deepseek V4.1 Flash (Deepseek) for completing a gap in an earlier proof, checking the validity of my proof and generating the format of this paper.

\end{document}